\documentclass[reprint,pra,nofootinbib,twocolumn]{revtex4-1}

\usepackage{bbm}
\usepackage{sarabian}
\usepackage{textcomp}
\usepackage{amsmath}
\usepackage{amsthm}
\usepackage{amsopn}
\usepackage{graphicx}
\usepackage{amsfonts}
\usepackage{amssymb}
\usepackage{mathrsfs}
\usepackage{mathbbol}
\usepackage{xfrac}

\newtheorem{theorem}{Theorem}

\newtheorem{remark}{Remark}

\newcommand{\sfo}{{\mathsf{O}}}

\newcommand{\sfa}{{\mathsf{A}}}
\newcommand{\sfb}{{\mathsf{B}}}

\newcommand{\ohat}{\hat{\omega}}

\usepackage{color}

\begin{document}

\title{Steering, incompatibility, and  Bell inequality violations in a class of probabilistic theories}

%\author{Neil Stevens\thanks{Email: {\tt ns695@york.ac.uk}}\hspace{3pt} and Paul Busch\thanks{Email:  {\tt paul.busch@york.ac.uk}}\\
%{\small Department of Mathematics, University of York, York, UK}}

\author{Neil Stevens}
\email{ns695@york.ac.uk}
\affiliation{University of York, York YO10 5DD, UK} 
\author{Paul Busch}
\email{paul.busch@york.ac.uk}
\affiliation{University of York, York YO10 5DD, UK}

\date{\today}

\begin{abstract}
We show that connections between a degree of incompatibility of pairs of observables and 
the strength of violations of Bell's inequality found in recent investigations can be extended to a general
class of probabilistic physical models. It turns out that the property of universal uniform steering is sufficient
for the saturation of a generalised Tsirelson bound, corresponding to maximal violations of Bell's inequality. 
It is also found that a limited form of steering is still available and sufficient for such saturation in some state spaces
where universal uniform steering is not given.
The techniques developed here are applied to the class
of regular polygon state spaces, giving a strengthening of known results. However, we also find indications that
the link between incompatibility and Bell violation may be more complex than originally envisaged.

\end{abstract}

\maketitle

%\vspace{1cm}

\section{Introduction}

The Bell inequalities \cite{Bell64} provide constraints that 
certain families of joint probability distributions must satisfy to admit a common joint distribution. It is known that
the satisfaction of a full set of Bell inequalities in a probabilistic system is equivalent to the existence 
of such a joint probability\cite{Fine82a,Fine82b}.\footnote{As observed by Pitowsky \cite{Pitowsky94}, Bell-type inequalities 
had already been formulated  as early as 1854 by George Boole, who deduced them as conditions for the 
possibility of objective experience \cite{Boole1854,Boole1862}.} It was observed subsequently that joint measurability 
(in the sense that there exist joint probabilities of the usual quantum mechanical form for every state) entails an 
operator form of Bell inequalities; therefore, the Bell inequalities are satisfied whenever the observables involved 
in an EPR-Bell type experiment are mutually commutative \cite{deMuynck84}. In the case of unsharp observables, 
commutativity is not required for joint measurability and the degree of unsharpness of the observables required for 
joint measurability can be determined; this value is more restrictive than is needed for violations of the Bell inequalities 
to be eliminated in the case of the singlet state \cite{Busch85c,Busch86,Kar95,Andersson05}. 

The connection between joint measurability and Bell inequalities -- in the specific form of the CHSH inequalities \cite{CHSH}, which
apply to experiments involving runs of measurements of two pairs of dichotomic observables on a bipartite system -- 
has been further elucidated in two interesting recent publications 
by Wolf {\em et al} \cite{Wolf09} and Banik {\em et al} \cite{Banik12}. The former have shown that for any pair of incompatible 
dichotomic observables in a finite dimensional quantum system a violation of a CHSH inequality will be obtained. 
Hence, incompatibility is not only necessary but also sufficient for obtaining Bell inequality violations. 
Wolf {\em et al} \cite{Wolf09} conclude that {``if a hypothetical no-signaling theory is a refinement of quantum mechanics 
(but otherwise consistent with it), it cannot render possible the joint measurability of observables which are incompatible 
within quantum mechanics".} With this result a tight link has been established between the availability of incompatible 
observables and the possibility of violating a CHSH inequality. It is natural to ask whether a quantitative connection 
can be found between  a degree of incompatibility and the strength of these violations, and whether such a connection
is specific to quantum mechanics or holds in a wider class of probabilistic physical theories.

It is a well known fact that two incompatible quantum observables can be {\em approximately} measured 
together if some unsharpness in the measurement is allowed. A measure of the incompatibility of two observables 
can then be obtained by quantifying the degree of unsharpness required to obtain an approximate joint measurement.
In the case of dichotomic observables this can be achieved by mixing each observable with a trivial observable
(a POVM whose positive operators are multiples of the identity)\footnote{Such mixing procedures and their connection 
with goal of achieving joint measurability are investigated systematically in \cite{BuHe08}).}, with relative weights $\lambda$, $1-\lambda$. 
The mixing weight determines the degree of unsharpness of the resulting smeared observable.

Banik {\em et al} have shown that the degree of incompatibility (they use the term complementarity) of two dichotomic 
observables, quantified by the largest smearing parameter, $\lambda$, for which the smeared versions are compatible, puts limitations 
on the maximum strength of  CHSH inequality violations available in such a theory \cite{Banik12}. The Bell functional, $\mathbb B$, a
generalisation of what is known as the Bell operator in the quantum case, then is
bounded by the parameter $\lambda_{\rm opt}$ associated with the ``most incompatible'' pair of observables, so that $\mathbb B\le2/\lambda_{\rm opt}$.
%Each such theory thus has its generalised Tsirelson bound.
%In order to make this modification the authors use a smearing with an unbiased trivial observable. 

Here we  study the connection between degrees of incompatibility and CHSH inequality violation in the context of
general probabilistic physical theories by way of unifying the approaches of \cite{Wolf09} and \cite{Banik12}. 
We will see that the degree of incompatibility used by Banik {\em et al} is closely linked 
with an unnamed parameter used in \cite{Wolf09} to characterise the joint measurability of two dichotomic observables.
Under an additional assumption on the physical theory, namely that it supports a sufficient degree of steering, 
the construction used to violate the CHSH inequality generalises. This gives a sufficient condition under which the 
maximal violation can be saturated. This result can be rephrased by saying that probabilistic theories 
can be classified according to the value of the {\em generalised Tsirelson bound}, defined as the maximum  
value of the Bell functional, and this bound  can (under said assumptions) be realised by suitable maximally incompatible
observables (see Theorem 1). 

Finally we illustrate the link between incompatibility and Bell violation in the class of regular polygon state spaces. It turns out
that this connection appears to hold generally in the case of even-sided polygons but not, at least in the same form, for 
odd-sided cases.

\section{General Probabilistic Models}

We begin by presenting the basic elements of the standard framework of probabilistic models. 
The framework was introduced in the 1960s by researchers in quantum foundations who used 
it to investigate axiomatic derivations of the Hilbert space formalism of quantum mechanics from 
operational postulates. Due to the emphasis on the convex structure of the set of states and the use of operations to 
model state transformations, the approach was called {\em convex state approach} or {\em operational approach}.
Some pioneering references are \cite{Ludwig67,Ludwig68,Mielnik69,Davies70,QTOS}. An overview 
of the literature and of relevant monographs can be obtained from \cite{SMQM} and \cite{OQP}. Recently the approach 
has gained renewed interest from researchers in quantum information exploring the information theoretic foundations 
of quantum mechanics. Accessible recent introductions can be found in e.g. \cite{Barnum06,Barnum09b,Pfister12}. 

The set of states $\Omega$ of a general probabilistic model  is taken to be a compact convex subset of a finite 
dimensional vector space $V$, where the convexity corresponds to the ability to define a preparation procedure as 
a probabilistic mixture of preparation procedures corresponding to other states. We write $A(\Omega)$ for the ordered 
linear space of affine functionals on $\Omega$, with the ordering given pointwise: 
$f \ge 0$ if $f(\omega) \ge 0$ for all  $\omega \in \Omega$. $A(\Omega)$ is also canonically an order unit space, 
with order unit $u$ defined by $u(\omega) = 1$ for all states $\omega \in \Omega$.
The (convex) set of effects on $\Omega$ is then taken to be the unit interval $[0,u]$ inside $A(\Omega)$, i.e.
\begin{equation}
\mathscr{E}(\Omega) = \{e \in A(\Omega) | 0 \le e(\omega) \le 1, \, \forall \, \omega \in \Omega \}.
\end{equation}
A discrete observable $\sfo$ is then a function from an outcome set $X$ into $\mathscr{E}(\Omega)$, 
that satisfies the normalisation condition $\sum_{x \in X} \sfo[x] = u$. The value (lying between $0$ and $1$) 
of $\sfo[x](\omega)$ denotes the probability of getting outcome $x$ for a measurement of the observable 
$\sfo$ in state $\omega$.

Under the assumption of tomographic locality \cite{Hardy11}, the state space of a composite system with local state spaces $\Omega_1$ and $\Omega_2$ 
naturally lives in the vector space $V_1 \otimes V_2$.
We then write $\Omega = \Omega_1 \otimes \Omega_2 = (V_1 \otimes V_2)_+^1$, where the normalisation is given 
by the order unit $u_1 \otimes u_2 \in V_1^* \otimes V_2^*$, but in general the positive cone is not unique \cite{NamPhelps1969}.

Although there is much choice in general for the ordering on $V_1 \otimes V_2$, there are two canonical choices, 
the \emph{maximal} and \emph{minimal}. As a minimal demand it is reasonable to expect $v_1 \otimes v_2 \ge 0$ 
whenever $v_1, v_2 \ge 0$, therefore we make the definition
\begin{equation}(V_1 \otimes_{min} V_2)_+ 
= \left\{\sum_{i,j} \lambda_{ij} v_1^{(i)}\otimes v_2^{(j)} \Big| \lambda_{ij} \in \mathbb{R}_+, v_k^{(i)} \in (V_k)_+\right\}.
\end{equation}
We can similarly make such demands on the order structure on $V_1^* \otimes V_2^*$ leading to the converse definition
\begin{equation}
(V_1 \otimes_{max} V_2)_+ = (V_1^* \otimes_{min} V_2^*)_+^*.
\end{equation}
Any cone on $V_1 \otimes V_2$ which lies between the maximal and minimal cones is then admissible as a viable 
order structure. In general the tensor product chosen is an important part in defining a theory; the only time when 
there is no choice (since maximal and minimal are the same) is when the local state spaces are simplexes
 \cite{NamPhelps1969}.
The case where both $\Omega_1$ and $\Omega_2$ are quantum state spaces provides a prime example of a nonminimal, 
nonmaximal order structure, namely the standard quantum mechanical tensor product. By definition 
$\Omega_1 \otimes_{min} \Omega_2$ contains only separable states, which form a proper subset of all bipartite states; 
by contrast, $\Omega_1 \otimes_{max} \Omega_2$ contains not only the usual quantum states, but also all 
normalised entanglement witnesses.

A bipartite state $\omega \in \Omega_1\otimes\Omega_2$ can also be viewed as a way to prepare states in $\Omega_1$, 
via the measurement of an observable on $\Omega_2$. In this way, for each state $\omega$, we can define the corresponding 
linear map $\ohat: V_2^* \rightarrow V_1$ by
\[
a(\ohat(b)) = \omega(a,b), \quad  a \in V_1^*, \ b \in V_2^*.
\]

\section{Fuzziness and joint measurability}

Consider a system represented by a probabilistic model, whose state space is given by the convex set $\Omega$. Any 
dichotomic (or two-outcome) observable $\sfo$ on $\Omega$ is determined by an effect 
$e=:\sfo[+1] \in \mathscr{E}(\Omega)$, where for any 
$\omega \in \Omega$, the probability of getting the outcome labelled by `+1' in the state $\omega$ is given 
by $e(\omega)$, and similarly for the outcome `-1' associated with the complement effect $e':=u-e=\sfo[-1]$.

Two effects $e$ and $f$ are said to be jointly measurable if there exists $g \in A(\Omega)$ satisfying
\begin{equation}
\begin{array}{c}0 \le g,\\
g \le e,\\
g \le f,\\
e+f \le g+u,\end{array}\label{JM}
\end{equation}
where $u$ is the order unit on $\Omega$. This condition is equivalent to the existence of a joint observable for the 
dichotomic observables corresponding to $e$ and $f$.

Given a two-outcome observable $\sfa$ determined by effect $e$, one can introduce a corresponding 
fuzzy observable $\sfa^{(\lambda)}$ as a smearing (or fuzzy version) of $\sfa$, whose defining effect is given by
\begin{equation}%\begin{array}{rcl}
e^{(\lambda)}  =  \frac{1+\lambda}{2}e + \frac{1-\lambda}{2}e' 
 =  \lambda e + \frac{1-\lambda}{2}u, \label{fuzz}
%\end{array}
\end{equation}
with smearing parameter $\lambda \in [0,1]$, and complement effect $e^{(\lambda)\prime} = e^{\prime(\lambda)}$.

Given any pair of two-outcome observables $\sfa_1, \sfa_2$, with corresponding effects $e, f$, we can use the parameter 
$\lambda$ to give a measure of how incompatible they are. First we note that for 
 %$\lambda = 0$ we have $e^{(0)} = f^{(0)} = \tfrac{1}{2}u$, so $e^{(0)}$ and $f^{(0)}$ are always trivially jointly measurable,
$\lambda=\frac12$, the choice of effect $g=\frac14(e+f)$ generates a joint observable for $e$ and $f$ since it satisfies  
\eqref{JM}, as is readily verified.
Thus the set of values of $\lambda$ which make $e^{(\lambda)}$ and $f^{(\lambda)}$ jointly measurable contains $\frac 12$. 
Further,  if $e^{(\lambda)}$ and $f^{(\lambda)}$ are jointly measurable, then for any $\lambda' \le \lambda$ 
so are $e^{(\lambda')}$ and $f^{(\lambda')}$. Hence the set lies inside the interval $[0,\lambda_{e,f}]$, 
where we define $\lambda_{e,f}$ to be the solution to the cone-linear program
\begin{align}\label{Fuzz}%\begin{array}{cc}
\textrm{maximise:\quad\quad} & \lambda \nonumber\\
\textrm{subject to:\quad}  g &\le e^{(\lambda)} \nonumber\\
 g &\le f^{(\lambda)} \\
 0 &\le g \nonumber \\
 e^{(\lambda)} + f^{(\lambda)} - u &\le g.\nonumber 
\end{align}
%\end{array}
%\end{equation}

This measure of incompatibility of a pair of effects in turn leads to a measure of the degree of incompatibility 
of a given model by looking for the most incompatible pair:
\begin{equation}\label{lambdaopt}
\lambda_{\rm opt} = \inf_{e,f \in \mathscr{E}(\Omega)} \lambda_{e,f}.
\end{equation}

Following a path similar to \cite{Wolf09}, we can define a different parameter $t_{e.f}$, which we will see 
is closely linked with $\lambda_{e,f}$. For a given pair of effects $e$ and $f$, we define $t_{e,f}$ to be the 
solution to the cone-linear program:
\begin{align}\label{Wolf}%\begin{array}{cc}
\textrm{minimise:\quad\quad} & t \nonumber\\
\textrm{subject to:\quad}  g &\le e + tu\nonumber \\
 g &\le f + tu \\
 0 &\le g \nonumber\\
 e + f - u &\le g. \nonumber
\end{align}
As shown in \cite{Beneduci13}, the optimal set for \eqref{Wolf} is nonempty, so the minimum can be achieved, hence $e$ and $f$ are incompatible if and only if $t_{e,f} > 0$. Here we notice 
that the pair $(\lambda,g)$ being feasible for the problem \eqref{Fuzz} is equivalent to the pair 
$\left(\tfrac{1-\lambda}{2\lambda},\tfrac{g}{\lambda}\right)$ being feasible for the problem \eqref{Wolf}. 
Combining this with the fact that the function $\tfrac{1-\lambda}{2\lambda}$ is monotonically decreasing 
for $\lambda \in [0,1]$ brings us to the promised link
\begin{equation}t_{e,f} = \frac{1-\lambda_{e,f}}{2\lambda_{e,f}}.\end{equation}

\subsection*{Examples}

In a model of discrete classical probability theory we take the state space to be the set of all probability measures 
on some countable set $X$, i.e.
\begin{equation} 
\Omega = \biggl\{ (\omega_x)_{x \in X}\, \big |\, \omega_x \ge 0  \ \forall x \in X,\ \sum_x \omega_x = 1  \biggr \}. 
\end{equation}
A functional $e$ on $\Omega$ with action $e(\omega) = \sum_x e_x\omega_x$ is easily seen to be positive 
iff $e_x \ge 0$ for all $x \in X$, and the order unit satisfies $u_x = 1$ for all $x \in X$.

Suppose we now have two effects $e,f  \in \mathscr{E}(\Omega)$. Taking $g$ to have components 
$g_x = \min\{e_x,f_x\}$, then since positivity is determined componentwise the inequalities \eqref{JM} are 
immediately satisfied, and hence $e$ and $f$ are jointly measurable. Since this holds for arbitrary $e$ 
and $f$ in this case we have $\lambda_{opt} = 1$.

%In classical probability where a state space is a simplex, all effects are jointly measurable \cite{?}, 
%thus in that case $\lambda_{opt} = 1$.

As shown in \cite{Banik12}, in any finite dimensional Hilbert space the value of the joint measurability 
parameter for a pair
of dichotomic observables is $\lambda_{\rm opt}=1/\sqrt2$.

A simple non-classical, non-quantum example is that of the {\em squit}. The two dimensional state space is given 
by a square, denoted $\square$; it contains all points $(x,y,1)$ with $-1 \le x+ y \le 1$, $-1 \le x-y \le 1$, and 
takes the shape of a square. As we will see, the squit leads to maximally incompatible effects in the sense 
that it leads to the 
smallest possible value of $\lambda_{opt}$.

Firstly we note that for any probabilistic model $\lambda = \tfrac{1}{2}$ provides a lower bound for $\lambda_{opt}$, 
since $e^{(\frac{1}{2})} = \tfrac{1}{2}e + \tfrac{1}{4}u$ and $f^{(\frac{1}{2})} = \tfrac{1}{2}f + \tfrac{1}{4}u$ 
are always jointly measurable. This can be seen explicitly by setting $g= \tfrac{1}{4}e+\tfrac{1}{4}f$, 
then the corresponding equations \eqref{JM} are satisfied.

As a convenient parametrisation we can write a generic affine functional $g\in A(\square)$ as a vector 
$g=(a,b,c)$, with action given by the canonical inner product scaled by a factor of $\frac{1}{2}$. In this 
case the order unit is given by $u=(0,0,2)$. Since the positivity of a functional $g$ on a compact convex 
set is equivalent to positivity 
on its extreme points, we can determine the structure of the set of effects by demanding that its elements $g$ take values 
between $0$ and $1$ on the extreme  points of the set of states. 
In the case of the squit, $\mathscr{E}(\square)$ is a convex polytope 
with defining inequalities given by
\begin{equation}
u\ge g\ge 0 \iff \left\{ \begin{array}{ll}    2\ge c+a\ge0, &\ 2\ge c+b\ge0,\\ 2\ge c-a\ge0,&\ 2\ge c-b\ge0.\end{array}\right.
\end{equation}
We note the extreme points: $(0,0,2)=u$, $(0,0,0)$, $(1,1,1)$, $(1,-1,1)$, $(-1,1,1)$, $(-1,-1,1)$.

In an attempt to find the lowest possible value of $\lambda_{e,f}$ we consider the case of the two orthogonal 
extremal effects $e=(1,1,1)$ and $f=(1,-1,1)$. 
In order for $e^{(\lambda)}$ and $f^{(\lambda)}$ to be jointly measurable we need to be able to find a 
$g$ that satisfies all the inequalities in \eqref{JM}. This entails, in particular:
\begin{align*}
g-e^{(\lambda)}-f^{(\lambda)}+u &= (a-2\lambda,b,c) \ge 0,\\
&\qquad\text{giving}\quad 2\lambda \le a+c ;\\
e^{(\lambda)}-g &= (\lambda-a,\lambda-b,1-c) \ge 0,\\
&\qquad\text{giving}\quad \lambda \le 1+a-c;\\
f^{(\lambda)}-g &= (\lambda-a,-\lambda-b,1-c) \ge 0,\qquad\\
&\qquad\text{giving}\quad \lambda \le 1-a-c;\\
g &= (a,b,c) \ge 0,\\
&\qquad\text{giving}\quad a \le c.
\end{align*}
Combining these inequalities leads to $4\lambda \le 2+a-c \le 2$, so for this choice 
of $e$ and $f$ we must have $\lambda_{e,f} \le \tfrac{1}{2}$. Given that 
$\frac{1}{2}$ is the lowest possible value, we conclude that in the case of the squit $\lambda_{opt} = \tfrac{1}{2}$.

\section{Steering and saturation of the generalised Tsirelson bound}

In order to give conditions on a generalised probabilistic model under which the bound  on CHSH violations
given in \cite{Banik12} can be achieved we need to introduce the notion of steering, as given in \cite{Barnum09a}.

Given two systems $A$ and $B$, with state spaces $\Omega_A$ and $\Omega_B$ respectively, for any bipartite 
state $\omega \in \Omega_A\otimes\Omega_B$ we can define its $A$ \emph{marginal}, living in $\Omega_A$ in 
an analogue to the quantum mechanical partial trace:
\begin{equation}
\omega^A = \ohat(u_B),
\end{equation}
where $u_B$ is the order unit on $B$, with a similar definition for $\omega^B$.

Following this we say that a state $\omega \in \Omega_A\otimes\Omega_B$ is \emph{steering} for its A marginal 
if for any collection of sub-normalised states that form a decomposition of that marginal, i.e., 
$\{\alpha_1,...,\alpha_n | \sum_i\alpha_i = \omega^A, 0 \le u_A(\alpha_i) \le 1 \}$, there exists an observable
$\{e_1,...,e_n\} \subset \mathscr{E}(\Omega_B)$ with $\alpha_i = \ohat(e_i)$.

It was observed by Schr\"odinger that this property holds in quantum mechanics for all pure bipartite states \cite{Schroedinger1936}, 
originally coining the term steering, which we generalise now, following \cite{Barnum09a}:
A general probabilistic model of a system $A$ with state space $\Omega_A$ supports \emph{uniform universal steering} 
if there is another system $B$ with state space $\Omega_B$, such that for any $\alpha \in \Omega_A$, there is a state 
$\omega_\alpha \in \Omega_A\otimes\Omega_B$, with $\omega_\alpha^A = \alpha$ that is steering for its $A$ 
marginal, and supports \emph{universal self-steering} if the above is satisfied with $B=A$.
The existence of steering in this manner is similar to the idea of purification to be found, for example,  in \cite{Chiribella10}. 
Indeed any purification of a state will be steering for its marginals; however steering states being pure is not required here.

The magnitude of maximal CHSH violations is quantified in quantum mechanics by the norm of the \emph{Bell operator}. 
We take $\sfa_1, \sfa_2, \sfb_1$ and $\sfb_2$ to be $\pm1$-valued 
observables, and define following \cite{Banik12} 
\begin{equation*}
\mathbb{B}:= \langle A_1B_1 + A_1B_2 + A_2B_1 - A_2B_2 \rangle_\omega  ,
\end{equation*}
where $A_1:=\sfa_1[+1]-\sfa_1[-1]$, etc., and $\langle X\rangle_\omega:=  X(\omega)$ for any affine functional $X$.
We will call the map $\omega\mapsto \mathbb{B}$ the  {\em Bell functional} and refer to 
$\sup_\omega \mathbb{B}$ as the (generalised) {\em Tsirelson bound}.
%We prove the following sharpening and generalisation of the result of \cite{Banik12}.

In order to see where steering enters the picture, we follow \cite{Banik12} to get a simple bound on the norm of $\mathbb{B}$. In order to do this we consider what effect smearing the observables of one party has by defining
\begin{equation}\mathbb{B}^{(\lambda)} = \langle A_1^{(\lambda)} B_1 + A_1^{(\lambda)} B_2 + A_2^{(\lambda)} B_1 - A_2^{(\lambda)} B_2 \rangle,\end{equation}
where $A_1^{(\lambda)}=\sfa_1^{(\lambda)}[+1]-\sfa_1^{(\lambda)}[-1]$etc., with the smearing of the effects as defined as in \eqref{fuzz}.
Due to the fact that the choice of observable that is mixed to form the smearing is an unbiased trivial observable, the resulting expectation scales with the smearing parameter:
\begin{equation}A_1^{(\lambda)}=\lambda\sfa_1[+1]+\frac{1-\lambda}{2}u-\lambda\sfa_1[-1]-\frac{1-\lambda}{2}u=\lambda A_1.\end{equation}
Now since the Bell functional is bilinear, and the same smearing parameter is being used on all functionals on the first system, the linear scaling carries over and we get $\mathbb{B}^{(\lambda)}=\lambda\mathbb{B}$.

As shown in the previous chapter, there always exists jointly measurable fuzzy versions of any pair of observables, so long as the value of the smearing parameter is small enough. Now if we take any $\lambda$ such that $A_1^{(\lambda)}$ and $A_2^{(\lambda)}$ are jointly measurable, then we know that the corresponding Bell functional satisfies the usual Bell inequality, and thus its value is bounded by $\mathbb{B}^{(\lambda)}\le 2$.
Consequently, each such value of $\lambda$ gives a bound on on the Bell functional of $\mathbb{B}\le\tfrac{2}{\lambda}$, and in order to obtain the lowest such upper bound we take the largest smearing parameter which still results in joint measurability, to get
\begin{equation}\mathbb{B}\le \frac{2}{\lambda_{\sfa_1[+1],\sfa_2[+1]}}.\end{equation}

Since every probabilistic model contains observables which are jointly measurable with no smearing, and thus satisfying the usual Bell inequality, knowing the above bound for a single pair of observables will not necessarily yield information about the structure of the system itself. A more general bound however can be written down by simply taking the most incompatible pair of observables:
\begin{equation} \mathbb{B}\le\frac{2}{\lambda_{\rm opt}}.\label{bound}\end{equation}

\begin{theorem}\label{thm1}
%In any probabilistic model of a system $A$ the Tsirelson bound is given by the tight inequality
In any probabilistic model of a systen $A$ that supports uniform universal steering, the Tsirelson bound is given by 
the tight inequality that can be saturated:
\begin{equation}
\mathbb{B} \, \le\, \frac2{\lambda_{\rm opt}},
\end{equation}
with $\lambda_{\rm opt}$ defined in Eq.~\eqref{lambdaopt}.
%If the probabilistic model supports uniform universal steering, then the bound in the above inequality
  %established in \cite{Banik} can be saturated.
\end{theorem}

\begin{proof}
Suppose we have a model of a system $A$ that supports uniform universal steering, and that we have two effects 
$e,f \in \mathscr{E}(\Omega_A)$. The parameter introduced earlier, $t_{e,f}$ can now also be calculated from the 
dual program to \eqref{Wolf}, which can be given as \cite{Convex}
\begin{align}\label{dual}%\begin{array}{cc}
\textrm{maximise:\quad}  \mu_3(e+f-u_A) &- \mu_1(e) - \mu_2(f) \nonumber\\
\textrm{subject to:\quad\ \ }  (\mu_1+\mu_2)(u_A) &= 1 \nonumber\\
 \mu_1+\mu_2 &= \mu_3+\mu_4 \\
% 0 
0&\le \mu_1, \mu_2, \mu_3, \mu_4\nonumber
\end{align}
with the $\mu_i \in A(\Omega_A)^*$.

Writing $\mu_1+\mu_2 = \rho$, for the $\mu_i$ that achieve the optimal value for \eqref{dual}, we find that 
$\rho \ge 0$ and $u_A(\rho)=1$, so $\rho \in \Omega_A$. By the assumption of uniform universal steering 
therefore we can find a state $\omega \in \Omega_A\otimes\Omega_B$ with $\omega^A =\ohat(u_B)= \rho$; 
moreover, in $\{\mu_1,\mu_2\}$ and $\{\mu_3,\mu_4\}$ we have two different decompositions of $\rho$, 
and we can thus find effects $\tilde{e},\tilde{f} \in \mathscr{E}(\Omega_B)$ satisfying
\begin{equation}
\ohat(\tilde{e}) = \mu_1, \qquad \ohat(\tilde{f}) = \mu_3.
\end{equation}

To achieve the maximum CHSH violations we take $\sfa_1, \sfa_2, \sfb_1$ and $\sfb_2$ to be $\pm1$-valued 
observables defined by effects $f',e,\tilde{e}'$ and $\tilde{f}'$ respectively; we then have
\begin{equation}\label{eqn:Bell-ops}
\begin{array}{cc}
A_1 = u_A - 2f,\ & B_1 = u_B - 2 \tilde{e}, \\
A_2 = 2e - u_A,\ & B_2 = u _B- 2\tilde{f}.
\end{array}
\end{equation}
The value of the Bell functional can now be evaluated:
\begin{align*}%\begin{array}{rcl}
%\langle A_1B_1 + A_1B_2 &+ A_2B_1 - A_2B_2 \rangle_\omega \\
\mathbb{B} & =  \omega(u_A-2f,2u_B-2\tilde{e}-2\tilde{f}) + \omega(2e-u_A,2\tilde{f}-2\tilde{e}) \\
& = 2\ohat(u_B)(u_A-2f) \\ &\qquad\qquad + 4\ohat(\tilde{e})(f-e) + 4\ohat(\tilde{f})(f+e-u_A) \\
& = 2+4[(\mu_1+\mu_2)(-f) \\ &\qquad\qquad + \mu_1(f) - \mu_1(e)  + \mu_3(f+e-u_A)] \\
& =  2+4[\mu_3(e+f-u_A) - \mu_1(e) - \mu_2(f)] \\
& =  2(2t_{e,f}+1)  =  \frac2{\lambda_{e,f}},
\end{align*}%\end{equation}
thus saturating the generalised Tsirelson bound as claimed.
\end{proof}

Not every probabilistic model may possess the property of supporting uniform universal steering, and although it is a sufficient 
condition to obtain the conclusion of the above theorem, as the following example will show, it is not a necessary 
one. Indeed a model of `boxworld', which contains Popescu-Rohrlich (PR) box states exhibiting the maximum possible CHSH 
violations, uses local state spaces that are the squits introduced earlier, and composition is given by the 
maximal tensor product. Despite the saturation of the generalised Tsirelson bound, such a state space does not admit 
uniform universal steering.

To see this, we consider a bipartite state $\omega \in \square\otimes_{max}\square$ with the corresponding map 
$\ohat$.  Note 
that from the definition of $\omega$ being a state, $\hat{\omega}$ will automatically be a positive map sending 
$V^*_+$ into $V_+$.
Now suppose $\omega$ is steering for its marginal $\rho$, i.e. $\hat{\omega}(u) = \rho$, and choose a 
decomposition of $\rho$ into pure states: $\rho = \sum_i \alpha_i$. Since the subnormalised states in the 
decomposition are pure, and $\hat{\omega}$ is positive, the inverse images $\hat{\omega}^{-1}(\alpha_i)$ 
must lie on extremal rays of the cone $V^*_+$. Consider the extremal ray effect $e=(1,1,1)$ with its complement
$e'=(-1,-1,1)$ (which is again extremal).
%For the case of the squit system the extremal rays are 
%generated by the vectors $(1,1,1)$, $(-1,-1,1)$, $(-1,1,1)$, and $(1,-1,1)$. If we denote the first of these $e$ 
%as before, we can see that its complement $e'$ is also extremal. 
With appropriate labelling of the $\alpha_i$ we can then 
write $\alpha_1 = \hat{\omega}(e)$ and $\alpha_2 = \hat{\omega}(e')$; however since we have 
$e+e'=u$,
\begin{equation*}
\alpha_1 +\alpha_2 = \hat{\omega}(e+e') = \hat{\omega}(u) = \rho,
\end{equation*}
and hence $\rho$ can be written as a mixture of just two pure states. Since there are many points in a 
square that can only be written as a convex combination of a minimum of three extreme points, we 
conclude that such a model of `boxworld' does not support universal uniform steering.

%{\bf SHORTEN the following into a remark on how the condition of homogeneity follows from UUS and how
%HOM and WSD entail the conditions of Thm 1 if the max tensor product is taken. A separate proof is then obsolete.}
\begin{remark}\rm 
 It is interesting to note that there is another set of conditions sufficient to obtain the conclusion of the 
 above theorem.   We say that a positive cone $V_+$ is \emph{homogeneous} if the space of order 
 automorphisms of $V$ acts transitively on the interior of $V_+$, and (weakly) \emph{self dual} if there 
 exists a linear map $\eta : V \rightarrow V^*$ that is an isomorphism of ordered linear spaces i.e. $\eta(V_+) = V_+^*$.
 It is known that homogeneity follows from uniform universal steering. Conversely,
if the positive cone $V_+$ generated by the state space $\Omega$ of the probabilistic model of a system 
$A$ is homogeneous and weakly 
self-dual, then uniform universal self-steering follows if the maximal tensor product is adopted. 
Hence the conditions of Theorem 1 are fulfilled \cite{Barnum09a} and the Tsirelson bound in the  inequality 
$\mathbb{B} \le 2/\lambda_{\rm opt}$  can be saturated.

In the quantum probabilistic model, the tensor product is not maximal but still uniform universal steering holds. The classical model (trivially)
satisfies the conditions of weak self-duality and homogeneity, and the tensor product is maximal. The squit is weakly self-dual but does not
satisfy uniform universal steering, so that homogeneity fails; but it allows enough self-steering so that the maximal Bell-Tsirelson 
bound of 4 can be realised.

\end{remark}

%\subsection*{Example}

%PR box system does not support uniform universal steering but nevertheless possesses maximally incompatible 
%pairs of effects, as seen above. Moreover, this system possesses a state that saturates the maximal Bell inequality 
%violation with $2/\lambda{\rm opt}=4$. {\bf (Detail?) ...............}

\section{Generalised Tsirelson bounds for polygon state spaces}

Work in \cite{Janotta} suggests that there is a spectrum of values for the generalised Tsirelson bound in the case of
2-dimensional polygon state spaces (given as the convex hulls of regular polygons). It is shown there that for a 
system composed of two identical polygon state spaces with an odd number of vertices, the maximally entangled
state does not lead to a violation of the standard Tsirelson bound of $2\sqrt2$, whereas in the case of an even number
of vertices this bound can be exceeded. This suggests that among the class of polygon state spaces, the generalised
Tsirelson bound can be either smaller or greater than the standard Tsirelson bound.

\begin{remark} We note that of the polygon state spaces, the only cases in which homogeneity holds are the $n=3$ triangle, and the $n\to\infty$ circle. Hence in general uniform universal steering is not available, however it may still be possible to saturate the generalised Tsirelson bound in some cases, but in other this may  not be possible.
\end{remark}

As shown in \cite{Janotta}, in the case of `boxworld', where each local state space is a square, the maximally entangled state is a PR box; it takes the maximum possible value for the Bell functional of $4$. This agrees with the result that the squit does indeed lead to the maximum amount of incompatibility, and shows that in this case the generalised Tsirelson bound can be saturated. We have been able to show that this conclusion holds also in regular polygon state spaces where the number of vertices is a multiple of 8. We expect this 
result to extend to all even-sided cases. This strengthens the expectation, expressed in \cite{Janotta}, that the in these cases the Tsirelson bound is saturated with the maximally entangled state.

Moving to the $n=5$ case makes things a lot more interesting however.
To see this we follow the notation in \cite{Janotta} and define the family of state spaces $\Omega_n$ to be the convex hull of the points
$$\omega_i = \left(\begin{array}{c}
r_n \cos(\tfrac{2\pi i}{n}) \\
r_n \sin(\tfrac{2\pi i}{n}) \\
1 \end{array}\right),
\qquad i=1,...,n$$
with $r_n=\sqrt{\sec(\tfrac{\pi}{n})}$.

The qualitative difference between the state spaces of odd and even sided polygons first appears in the structure of the set of effects. For the case of even $n$, along with $0$ and $u$, there are $n$ extremal effects:
$$e_i = \frac{1}{2}\left(\begin{array}{c}
r_n \cos(\tfrac{(2i-1)\pi}{n}) \\
r_n \sin(\tfrac{(2i-1)\pi}{n}) \\
1 \end{array}\right),
\qquad i=1,...,n$$
and in this case all the $e_i$ lie on extremal rays of the cone $V_+^*$. This important fact occurs since for each of the $e_i$ we can find another effect $e_j$, also extremal, which is it's compliment, i.e. $e_j = e_i' = u-e_i$, namely for $j=i+\tfrac{n}{2} \, \textrm{mod} \, n$.
For the case of odd $n$, a seemingly similar expression arises for the ray extremal effects:
$$e_i = \frac{1}{1+r_n^2}\left(\begin{array}{c}
r_n \cos(\tfrac{2\pi i}{n}) \\
r_n \sin(\tfrac{2\pi i}{n}) \\
1 \end{array}\right),
\qquad i=1,...,n$$
On this occasion however, the compliments of the $e_i$ are given by
$$e_i' = u-e_i = \frac{1}{1+r_n^2}\left(\begin{array}{c}
-r_n \cos(\tfrac{2\pi i}{n}) \\
-r_n \sin(\tfrac{2\pi i}{n}) \\
r_n^2 \end{array}\right),
\qquad i=1,...,n$$
which do not coincide with the $e_i$, and thus there are $2n$ non-trivial extreme points of $\mathscr{E}(\Omega_n)$.

Now we can pose the question of what the value is for $\lambda_{opt}$ when the state space is $\Omega_5$, and whether is it possible to achieve the corresponding Bell value $\mathbb{B} = 2/\lambda_{opt}$.
Since each extreme two valued observable is determined by a ray effect, the largest value of incompatibility will come from one of the possible pairs of the $e_i$. However due to the symmetry of the state space, the affine transformation of rotating by $\pi/5$ serves only to cyclically permute the indices of the $e_i$ modulo $5$. This means that there are only two possible values of $\lambda_{e_i,e_j}$, those for nearest neighbors, and those for next nearest neighbors.
Calculation shows that these values are, for example
$$\lambda_{e_1,e_2} = \frac{3+2\sqrt{5}}{11} \approx 0.67928,$$
$$\lambda_{e_1,e_3} = \frac{8+3\sqrt{5}}{19} \approx 0.77416.$$
hence the value of $\lambda_{opt}$ for the pentagon is $\tfrac{3+2\sqrt{5}}{11}$.
From \eqref{bound} this gives the bound on the Bell functional as $\mathbb{B}\le 4\sqrt{5}-6$, however unlike in the case of the tensor product of two squits, the maximally entangled state between two pentagonal state spaces does not saturate the corresponding bound; instead we get a value of $\mathbb{B} = \tfrac{6}{\sqrt{5}}$, strictly below that coming from the level of incompatibility on one state space. This fact suggests that either the chosen way of evaluating the level of incompatibility in a system used does not capture everything, or that there is some structural obstruction that prevents such a link holding, that does not exist on other cases. Here we present some evidence towards the former.

In order to improve the measure of incompatibility used, we wish to modify the program used in eqn. \eqref{Fuzz}. To do this we relax the method of smearing used, still mixing in multiples of the order unit, corresponding to trivial observables; but we now allow them to  be possibly biased as follows:
\begin{equation} e^{(\lambda,p)} = \lambda e + p(1-\lambda) u.\end{equation}
This definition encompasses the old, with $e^{(\lambda)} = e^{(\lambda,\frac12)}$.

The updated measure of incompatibility of a pair of effects $e$ and $f$, which we denote $\bar{\lambda}_{e,f}$, is now given by the optimal value of the optimisation program
\begin{align}\label{Fuzz2}
\textrm{maximise:\quad\quad} & \lambda \nonumber\\
\textrm{subject to:\quad}  g &\le e^{(\lambda,p)} \nonumber\\
 g &\le f^{(\lambda,q)} \\
 0 &\le g \nonumber \\
 e^{(\lambda,p)} + f^{(\lambda,q)} - u &\le g\nonumber \\
0 &\le p,q \le 1. \nonumber
\end{align}

Solving this updated problem in the case of the pentagon again gives the optimal value on e.g. $e_1$ and $e_2$, with
$$\bar{\lambda}_{opt} = \frac{5+\sqrt{5}}{10} \approx 0.72361,$$
which occurs for the values $p=q=1$.

This is indeed a different value from earlier, but still we have that $\frac{2}{\bar{\lambda}_{opt}} \ne \frac{6}{\sqrt{5}}$, however in this case, the unbiased nature of the observables mixed in means such a simple link is no longer expected, and indeed we can see that there is a link to the Bell value on the maximally entangled state as follows.
As in the previous, we can define a smeared version of the Bell functional, where the smearing is all done on the functionals of one party:
\begin{equation}\mathbb{B}^{(\lambda,1)} = \langle A_1^{(\lambda,1)} B_1 + A_1^{(\lambda,1)} B_2 + A_2^{(\lambda,1)} B_1 - A_2^{(\lambda,1)} B_2 \rangle,\end{equation}
but now instead of having the nice linear scaling in $\lambda$, we gain an extra expectation term $\mathbb{B}^{(\lambda,1)} = \lambda\mathbb{B} + 2(1-\lambda)\langle B_1 \rangle$, and again under the assumption that $\lambda$ is small enough to ensure joint measurability, and then taking the largest such value we can write the inequality
\begin{equation}\mathbb{B}\le\frac{2\bigl[1-(1-\bar{\lambda}_{opt})\langle B_1 \rangle\bigr]}{\bar{\lambda}_{opt}}.\label{Bell2}\end{equation}
The link to the maximally entangled state on two pentagons now comes from noting that the expectation of any observable $B_1$ defined by an extreme effect on the maximally entangled state is $\langle B_1 \rangle = \tfrac{5-2\sqrt{5}}{5}$.
This means that if evaluated in the maximally entangled state, the inequality in \eqref{Bell2}, for the value of $\bar{\lambda}_{opt}$ given above, is indeed saturated.

\section{Conclusion}

By combining and developing ideas  from the works of Wolf {\em et al} \cite{Wolf09} and Banik {\em et al} \cite{Banik12},
we have shown that probabilistic models can be classified according to their associated value of the generalised Tsirelson 
bound, which specifies the maximum possible violation of CHSH inequalities. We have given conditions (defined and
studied in  \cite{Barnum09a}), that probabilistic models may or may
not satisfy, under which the maximal CHSH violations are attained for appropriate choices of maximally incompatible
dichotomic observables. Here the degree of the incompatibility of two observables is defined by the minimum amount of
smearing of these observables necessary to turn them into jointly measurable observables.

The authors of  \cite{Wolf09} concluded that observables that are incompatible in quantum mechanics remain incompatible in any probabilistic 
model that serves as an extension of quantum mechanics. Here we have shown that this conclusion applies to extensions of any probabilistic
model that allows for sufficient steering.

As an illustration of the general results we have considered the squit system which underlies the PR box model, 
and have identified the pair of maximally incompatible extremal effects of the squit that give rise to the saturation 
of the largest possible value (i.e., 4) of the Tsirelson bound. In addition, we have obtained partial confirmation of the
conjectured maximality of the Bell functional if evaluated on the maximally entangled state in the class of regular polygon state spaces
considered in \cite{Janotta}.

In the case of the pentagon state space we discovered that the connection between incompatibility and Bell violation is not 
always of the simple form envisaged originally and used through most of this paper; this suggests that the definitive universal
expression of this connection remains yet to be found.

The methods used here are taken from amongst some of the  standard tools of quantum measurement and 
information theory used in
\cite{Wolf09} and \cite{Banik12}, and we have shown that they apply equally well in a wide class of probabilistic models.
This insight may prove valuable in future investigations into the characterisation of quantum mechanics among all
probabilistic models.

\section*{Acknowledgements} We wish to thank Takayuki Miyadera for the suggestion to consider the generalisation of the 
optimisation problem for the incompatibility parameter that allows mixing with trivial observables that are not necessarily unbiased.
We are also grateful to Manik Banik and Alexander Wilce for helpful comments on a draft version of the paper. N.S. gratefully acknowledges support through the award of an Annie Currie Williamson PhD Bursary at the University of York.

%\bibliography{bib-PB-20130226}

%\bibliographystyle{unsrt}

%\end{document}

\end{document}